\documentclass[a4paper, USenglish, cleveref, autoref, thm-restate]{lipics-v2021}

\usepackage{float}  
\usepackage{mdframed}
\usepackage{xspace}

\usepackage[linesnumbered,ruled,vlined]{algorithm2e}
\usepackage[noend]{algpseudocode}

\newtheorem{fact}[theorem]{Fact}

\title{A Simple Boosting Framework for Transshipment} 

\author{Goran Zuzic}{Google Research, Switzerland \and ETH Z\"urich, Switzerland\and \url{https://goranzuzic.github.io/} }{goranzuzic@google.com}{https://orcid.org/0000-0002-9322-6329}{Work presented in this paper was partially performed while at ETH Z\"urich. Supported in part by NSF grants CCF-1814603, CCF-1910588, NSF CAREER award CCF-1750808, a Sloan Research Fellowship, and funding from the European Research Council (ERC) under the European Union's Horizon 2020 research and innovation program (grant agreement No. 949272).}

\authorrunning{G. Zuzic} 

\Copyright{Goran Zuzic} 



\category{} 

\relatedversion{https://arxiv.org/abs/2110.11723} 

\acknowledgements{The author would like to thank Bernhard Haeupler, Patrik Pavic, and Richard Peng for helpful discussions about the paper. The author would also like to thank the anonymous reviewers for their helpful suggestions that significantly improved the quality of the paper.}

\nolinenumbers 

\EventEditors{Inge Li G{\o}rtz, Martin Farach-Colton, Simon J. Puglisi, and Grzegorz Herman}
\EventNoEds{4}
\EventLongTitle{31st Annual European Symposium on Algorithms (ESA 2023)}
\EventShortTitle{ESA 2023}
\EventAcronym{ESA}
\EventYear{2023}
\EventDate{September 4--6, 2023}
\EventLocation{Amsterdam, the Netherlands}
\EventLogo{}
\SeriesVolume{274}
\ArticleNo{73}

\newcommand{\norm}[1]{\left\lVert#1\right\rVert}

\newcommand{\poly}{\mathrm{poly}}

\newcommand{\eps}{\varepsilon}

\newcommand{\R}{\mathbb{R}}
\newcommand{\Z}{\mathbb{Z}}

\newcommand{\dist}{\mathrm{dist}}

\newcommand{\polylog}{\mathrm{polylog\ }}

\newcommand\twovec[2]{\begin{bmatrix}#1 \\ #2\end{bmatrix}}

\newcommand{\1}[1]{\mathbb{I}[#1]}

\newcommand{\opt}{\mathrm{OPT}}

\newcommand{\inner}[1]{\left\langle #1 \right\rangle}
\newcommand{\residual}{\ensuremath{_{\mathrm{residual}}}\xspace}
\DeclareMathOperator{\smax}{smax}

\begin{document}

\maketitle

\begin{abstract}
  Transshipment is an important generalization of both the shortest path problem and the optimal transport problem. The task asks to route a demand using a flow of minimum cost over (uncapacitated) edges. Transshipment has recently received extensive attention in theoretical computer science as it is the centerpiece of all modern theoretical breakthroughs in parallel and distributed (approximate) shortest-path computation, a classic and well-studied problem.

  \medskip

  The key advantage of transshipment over shortest paths is the so-called \emph{boosting} property: one can often boost a crude approximate solution to a (near-optimal) $(1+\varepsilon)$-approximate solution. However, our understanding of this phenomenon is limited: it is not clear which approximators can be boosted. Moreover, all current boosting frameworks are built with a specific type of approximator in mind and are relatively complicated.

  \medskip

  The main takeaway of our paper is conceptual: \emph{any} black-box oracle that computes an approximate \emph{dual} solution can be boosted to an $(1+\varepsilon)$-approximator. This decouples and simplifies all known near-optimal $(1+\varepsilon)$-approximate transshipment and shortest paths results: they all (implicitly) construct approximate dual solutions and boost them.

  \medskip

  We provide a very simple analysis based on the multiplicative weights framework. Furthermore, to keep the paper completely self-contained, we provide a new (and arguably much simpler) analysis of multiplicative weights that leverages well-known optimization tools to bypass the ad-hoc calculations used in the standard analyses.
\end{abstract}

\section{Introduction}

Suppose we are given a weighted graph $G = (V, E)$ and a \emph{demand vector} $d \in \R^V$ satisfying $\sum_{v \in V} d(v) = 0$, where $d(v)$ denotes the number of units of some (single) commodity that are either available (if $d(v) > 0$) or required (if $d(v) < 0$) at the node $v$. \emph{Transshipment} asks to distribute the available units of the commodity until it perfectly matches the requirement. The goal is to minimize the total cost of movement, where moving a single unit over an edge $e$ has a cost of $w(e)$ ($w(e) \ge 0$ is the weight of $e$).

Transshipment is a strong generalization of multiple problems, including the $s-t$ shortest path problem, optimal transport, or the assignment problem on metric spaces.
\begin{example}
  Given two nodes $s, t$ in a weighted graph $G$, the shortest path can be modeled as transshipment for the demand $d(v) := \1{v = s} - \1{v = t}$.  
\end{example}
\begin{example}
  Let $(V, d)$ be a metric space and let $A, B \subseteq V$ be two disjoint subsets. The minimum cost perfect matching between $A$ and $B$ (the so-called assignment problem~\cite{ramshaw2012minimum}) can be modeled as transshipment on the complete bipartite graph $(A \cup B, A \times B)$ with weights $w(\{a, b\}) = d(a, b)$ and the demand vector $d(v) := \1{v \in A} - \1{v \in B}$. This also models the Wasserstein distance (also known as the earth mover's distance or optimal transport~\cite{villani2009optimal}).
\end{example}

Perhaps surprisingly, transshipment has proven to be extremely useful for uncovering the distance structure (i.e., shortest paths) of a graph. Indeed, the problem has been the centerpiece of all near-optimal approaches to the single-source shortest path problem in the parallel and distributed settings~\cite{BKKL17, Li20, AndoniSZ20, goranci2022universally, rozhon2022undirected}.

The key property that differentiates transshipment from other similar problems like shortest path is the so-called \emph{boosting} property---one can boost a crude, say $\poly(\log n)$-approximate solution, to a near-optimal $(1 + \eps)$-approximate solution. This conceptually reduces $(1+\eps)$-transshipment (and shortest path computation) to approximate transshipment. However, not all approximators can be boosted and a more principled understanding of which approaches are susceptible to boosting is required. Moreover, the current boosting algorithms are coupled together with the specific approximators they use, making them non-modular, complicated, and hard to reuse.

The main takeaway of our paper is conceptual: \emph{any} black-box oracle that computes a $\alpha$-approximate \emph{dual} solution can be boosted to a $(1 + \eps)$-approximate dual solution. This significantly simplifies current transshipment results by decoupling them into two independent questions: (1) how to obtain an approximate dual solution (which is often model-specific), and (2) how to boost this approximate solution (which can be reused). The scope of this paper is to develop a simple framework for the latter question.

We provide a very simple algorithm and analysis based on the \emph{multiplicative weights framework}. Furthermore, to keep the paper completely self-contained, we provide a new (and arguably much simpler) analysis of multiplicative weights that leverages well-known optimization tools to bypass the ad-hoc calculations used in the standard analyses. (\Cref{sec:simple-mw-analysis}).

We now provide several examples that show how prior approaches all (implicitly) construct approximate dual and then boost them.

\begin{example}Sherman~\cite{She17b} gave the first sequential almost-linear\footnote{We refer to $m \poly(\log n)$ as near-linear and $m^{1+o(1)}$ as almost-linear runtimes.} $(1+\eps)$-transshipment algorithm. The main insight was the construction of a so-called \emph{linear cost approximator}, which is a linear operator $R$ (i.e., matrix) such that $\norm{R d}_1$ approximates the optimal transshipment cost in the sense that $\opt(d) \le \norm{R d}_1 \le n^{o(1)} \opt(d)$ for all demands $d \in \R^V$. Their paper uses linear cost approximators with subgradient descent to show one can obtain a $(1+\eps)$-approximate solution. We provide a conceptual decoupling and reinterpretation of their paper: one can use any linear cost approximator $R$ to directly obtain an approximate dual solution, which can, in turn, be boosted to an $(1+\eps)$-approximate solution via our framework.
\end{example}

\begin{example}Haeupler and Li~\cite{HL18} solve $n^{o(1)}$-approximate transshipment in the distributed setting and leave the possibility of boosting to an $(1+\eps)$-approximation as the main open problem, which would have yielded important consequences in the distributed setting. Our paper provides a partial explanation to why their approach was not susceptible to boosting: their approach, based on low-stretch spanning trees, only computes a \emph{primal} solution (i.e., an approximate flow), whereas a dual solution is required. A dual-based solution was later recently developed by Rozhon et al.~\cite{rozhon2022undirected}.
\end{example}

\begin{example}Other successful $(1+\eps)$-transshipment approaches either approximate the solution by solving the original problem on a spanner~\cite{BKKL17}, or by constructing a linear cost approximator on an emulator of the original graph~\cite{Li20, AndoniSZ20} (an emulator of $G$ is a graph $H$ whose distance structure multiplicatively approximates the one of $G$, see \Cref{sec:application-spanners-emulators}). We show all of these approaches can be reinterpreted as obtaining an approximate dual solution.
\end{example}

\textbf{Comparison with Becker et al.~\cite{BKKL17}.} The paper contributed the first polylog-competitive existentially-optimal shortest path algorithm in the distributed setting (up to $\tilde{O}(1)$-factors). Crucially, they develop a boosting framework for transshipment which, similar to this paper, uses an approximate dual solver to construct a near-optimal solution. The main drawbacks of their solver are that (1) the analysis of \cite{BKKL17} is quite involved, stemming from it being based on projected gradient descent, and (2) as written, the interface of the \cite{BKKL17} solver relies on solving a modified version of transshipment which is harder to interpret and work with than the original one. As stated in the journal version of \cite{BKKL17}, their interface can be significantly simplified (by working with projections), but this degrades the runtime to have an $\alpha^4$-dependency w.r.t. the approximation quality $\alpha$ of the approximator (we provide an $\alpha^2$-dependency) and requires non-explicit modifications to the solver that might be difficult for non-experts.\footnote{Specifically, they require the returned dual solution be orthogonal to the demand vector. However, as they note in the journal version of their paper, this issue can be mitigated by projecting a general solution to the space of vectors orthogonal to the demand: this can be shown to work with a loss in approximation factor and time complexity if one appropriately initializes the solution.} On the other hand, the approach presented in our paper has several drawbacks compared to \cite{BKKL17}, such as: (1) our solver requires a guess on the optimal solution, which is obtained using binary search, while their solver does not need adapting the internal parameters during the optimization process, and (2) our dual-only solver needs to perform extra steps to return a feasible primal solution.
However, independent of the drawbacks, we believe the user-friendly interface, better runtime, and a simpler analysis make our conceptual contribution worthwhile.

\medskip

\textbf{Potential impact.} Ultimately, we hope that this paper will encourage an ongoing effort to simplify deep algorithmic results that use continuous optimization tools. Such an effort would potentially yield a dual benefit: it would both lower the barrier to entry for newcomers by conceptually simplifying the current approaches, as well as help to transfer the modern theoretically-optimal algorithms into real-world state-of-the-art by allowing practitioners to independently combine the theoretical ideas with the many heuristics necessary for an algorithm to perform well in practice.

\medskip

\textbf{Organization of the paper.} We present a model-oblivious boosting framework for transshipment in \Cref{sec:boosting-framework} and apply it in \Cref{sec:applications} to simplify previous results. These applications are loosely grouped by the method of computing the approximate dual solution: \Cref{sec:application-spanners-emulators} presents results when the approximate solution is computed on a spanner or emulator (i.e., on graphs that approximate the original metric). \Cref{sec:linear-cost-approximators} presents results that compute the dual solution via (aforementioned) linear cost approximators. Finally, \Cref{sec:simple-mw-analysis} gives a simple and self-contained analysis of multiplicative weights.

\section{Preliminaries}\label{sec:prelims}

\textbf{Graph Notation.} Let $G = (V, E)$ be a undirected graph and let $n := |V|, m := |E|$.
It is often convenient to direct $E$ consistently. For simplicity and without loss of generality, we assume that $V = \{v_1, v_2, \ldots, v_{n}\}$ and define $\vec{E} = \{ (v_i,v_j) \mid (v_i, v_j) \in E, i < j\}$. We identify $E$ and $\vec{E}$ by the obvious bijection. We chose this orientation for simplicity and concreteness: arbitrarily changing the orientations does not influence the results (if done consistently).
We denote with $B \in \{-1,0,1\}^{V \times \vec{E}}$ the node-edge incidence matrix of $G$, which for any $v \in V$ and $e = (s,t) \in \vec{E}$ assigns $B_{s,e} = 1$, $B_{t,e} = -1$, and $B_{u,e} = 0$ when $u \not \in \{s, t\}$.
A weight or length function $w$ assigns each edge $e \in \vec{E}$ a weight $w(e) > 0$. The weight function can also be interpreted as a diagonal \emph{weight matrix} $W \in \R_{\ge 0}^{\vec{E} \times \vec{E}}$ which assigns $W_{e,e} = w(e) \ge 1$ for any $e \in \vec{E}$ (and $0$ on all off-diagonal entries).
In this paper, it is often more convenient to specify weighted graphs via $G \cong (B, W)$, i.e., by specifying its matrices $B$ and $W$ as defined above.

\begin{figure}[h]
   \centering
   \includegraphics[width=0.5\textwidth,trim=0 0 0 0]{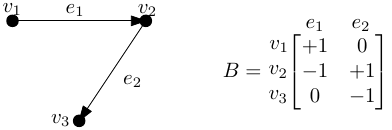}
   \caption{A simple graph $G$ and its corresponding node-edge incidence matrix $B$.}
 \label{fig:example1}
\end{figure}

\textbf{Flows and Transshipment (TS).}\label{sec:prelim-transshipment}
%
A \emph{demand} is a $d \in \R^{V}$. We say a demand is \emph{proper} if $\sum_{v \in V} d_v = 0$. A \emph{flow} is a vector $f \in \R^{\vec{E}}$, where $f({\vec{e}}) > 0$ if flowing in the direction of the arc $\vec{e}$ and negative otherwise. A flow $f$ \emph{routes} a demand $d$ if $Bf = d$. It is easy to see only proper demands are routed by flows. The cost of a flow $f$ is $\norm{Wf}_1$.
For a weighted graph $G$ and a given proper demand $d$ the \emph{transshipment problem} (or TS, for short) asks to find a flow $f_d^*$ of minimum-cost among flows that route $d$. In other words, the tuple $(B, W, d)$ specifies a transshipment instance. When the underlying graph $G \cong (B, W)$ is clear from the context, we define $\norm{d}_\opt := \norm{ W f_d^*}_1$ to denote the cost of the optimal flow for routing demand $d$.
%
%
The transshipment problem naturally admits the following LP formulation and its dual. The primal asks us to optimize over all \emph{flows} $f \in \R^{\vec{E}}$, while the dual asks us to optimize over all vectors $\phi : \R^V$, which we refer to as \emph{potentials}.
\begin{align}
  \textbf{Primal:}~~ \min~ \norm{Wf}_1 : Bf = d, && \textbf{Dual:}~~ \max~ \inner{d, \phi} : \norm{W^{-1} B^\top \phi}_{\infty} \leq 1. \label{eq:TS-primal-dual}
\end{align}
Scalar products are denoted as $\inner{x, y} = x^T \cdot y$. Finally, we assume the weights and demands are polynomially-bounded, hence $\norm{d}_\opt \le n^{O(1)}$. Any feasible primal and dual values provide an upper and lower bound on $\norm{d}_\opt$, formally stated in the following well-known result.
\begin{fact}\label{fact:weak-duality}
  Let $f \in \R^{\vec{E}}$ and $\phi \in \R^V$ be any feasible primal and dual solution: i.e., $Bf = d$ and $\norm{W^{-1} B^T \phi}_\infty \le 1$. Then $\inner{d, \phi} \le \norm{d}_\opt \le \norm{W f}_1$.
\end{fact}

For example, consider the $s$-$t$ shortest path subproblem where $d(v) = \1{v = s} - \1{v = t}$ specifies the demand. One optimal solution to the primal/dual pair is to set $f(\vec{e})$ to $1$ iff $\vec{e}$ is on some fixed shortest path from $s$ to $t$; $\phi(v)$ is set to the distance in $G$ from $t \in V$. Note that in this case the primal and dual objectives are equal, and correspond to the weight of the shortest path from $s$ to $t$.

\begin{figure}[h]
   \centering
   \includegraphics[width=\textwidth,trim=0 0 0 0]{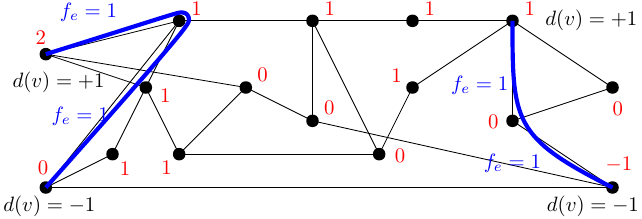}
   \caption{A example transshipment graph with its exact solution. The original graph is unit weight $w_e = 1$ and undirected. The demand $d$ is non-zero at four nodes. The optimal primal flow $f$ is depicted in \textcolor{blue}{blue} and is non-zero for four edges. One of many optimal vectors $\phi$ is depicted in \textcolor{red}{red}. The optimal value of the solution is $\opt = 4$.}
 \label{fig:example2}
\end{figure}

\textbf{Asymptotic Notation.}
We use $\tilde{O}$ to hide polylogarithmic factors in $n$, i.e., $\tilde{O}(1) = \polylog n$.

\textbf{Algorithmic model and basic vector operations.} To facilitate both simplicity and generality, we specify our algorithms using high-level operations. Specifically, in a unit operation, we can perform the following so-called \textbf{basic vector operations}: (1) assign vectors in $\R^n$ or $\R^m$ to variables, (2) add two (vector) variables together, (3) apply any scalar function $\lambda : \R \to \R$ to each component of a vector separately, and (4) compute matrix-vector products with matrices $B$, $B^T$, $W$, and $W^{-1}$. Note that each basic vector operation can be near-optimally compiled into standard parallel/distributed models. In PRAM: each operation can be performed in $\tilde{O}(1)$ depth and near-linear work. In the standard distributed model of computation CONGEST~\cite{peleg2000distributed} basic vector operations can be computed in a single round of distributed computation (where the variables are stored in the obvious distributed fashion).

\textbf{Multiplicative weights (MW) framework} is a powerful meta-algorithm that allows for (among other things) solving various optimization tasks by reducing them to simpler (so-called ``linearized'') versions of the original problem~\cite{AroraHK12}. For the purposes of this paper, we define the following pair of tasks.
\begin{definition}\label{table:feasbility-task-def}
  Let $\gamma \in \mathbb{R}$ be a scalar, $A \in \R^{m \times n}$ be a matrix, and $b \in \R^n$ be a vector. We define the following \textbf{Feasibility task}:
  \begin{align}
    \exists? x\in\R^n \mid\quad \norm{Ax}_\infty + \inner{b, x} \le \gamma . \label{eq:feas-task}
  \end{align}

  Given additionally a vector $p \in \R^m$ satisfying $\norm{p}_1 \le 1$ and an accuracy parameter $\eps > 0$, we also define the following \textbf{Linearized task}:
  \begin{align}
    \exists? x\in\R^n \mid\quad \inner{p, Ax} + \inner{b, x} \le \gamma - \eps . \label{eq:linear-task}
  \end{align}
\end{definition}
Note that if the Feasibility task (\Cref{eq:feas-task}) is feasible, then the Linearized task (\Cref{eq:linear-task}) is also feasible for every $p$ (satisfying $\norm{p}_1 \le 1$) and every $\eps > 0$.



Suppose we want to solve some fixed Feasibility task (\Cref{eq:feas-task}) and assume we know how to solve the accompanying (typically much easier!) Linearized task (\Cref{eq:linear-task}, for any $p$ and $\eps$) via some \emph{black-box Oracle}. Then, there exists a simple algorithm that computes a solution to the Feasibility task by repeatedly querying the Oracle with different values of $p$ that satisfy $\norm{p}_1 \le 1$ (the accuracy $\eps$ stays fixed); the oracle is assumed to return a feasible solution $x$ for each queried Linearized task. 



We define the \textbf{width of the Oracle} $\rho > 0$ to be (any upper bound on) the largest width of a solution $\norm{Ax}_{\infty}$ that can be returned by the Oracle, i.e., $\rho \ge \norm{A x}_{\infty}$ during the course of the algorithm. Oracles with larger widths need to be queried more times, hence we aim to construct Oracles with their width being as small as possible. The following \Cref{theorem:MW-feasibility} and \Cref{algo:MW-strategy-feasibility} give a solver for the Feasibility task (\Cref{eq:feas-task}) assuming the Oracle. We defer the proof to \Cref{sec:other-MW-forms}.

\begin{theorem}\label{theorem:MW-feasibility}
  Let $(A, b, \gamma)$ be a feasible Feasibility task (\Cref{eq:feas-task}) and fix $\eps > 0$. Suppose we have access to an Oracle that will solve the accompanying Linearized task \Cref{eq:linear-task} specified by $(A, b, \gamma, p, \eps)$ for any $\norm{p}_1 \le 1$. Then, \Cref{algo:MW-strategy-feasibility} constructs a feasible solution for \Cref{eq:feas-task} and queries the Oracle at most $4 \eps^{-2} \rho^2 \ln (2m)$ times. Here, $\rho > 0$ is the width of the Oracle.
\end{theorem}

\begin{algorithm}[h]
  \caption{Solver for the Feasibility task using an oracle for the Linearized task.}
  \label{algo:MW-strategy-feasibility}
  \begin{enumerate}
  \item \textbf{Input:} Feasbility task ($A \in \R^{m \times n}, b \in \R^n, \gamma \in R$) and $\eps > 0$.
  \item Initialize $x_* \gets \vec{0} \in \R^n$ and $\beta := \eps / (2 \rho^2)$.
  \item For $t = 1, \ldots, T$ rounds, where $T := 4 \eps^{-2} \rho^2 \ln (2m)$:
    \begin{enumerate}
    \item Let $q \gets \begin{bmatrix}A \\ -A\end{bmatrix} x_* \in \R^{2m}$.
    \item Let $q'_i \gets \exp(\beta q_i)_i$ for $i \in [2m]$.
    \item Let $p_t \gets (1 / \sum_{i=1}^{2m} q'_i)(q'_i - q'_{i+m})$. (Normalization and flattening.)
    \item MW outputs $p_t \in \R^m$ to Oracle. (Note that $\norm{p_t}_1 \le 1$.)
    \item Oracle returns a solution $x_t \in \R^n$ to the Linearized task w.r.t. $p_t$. ($\rho$ must be set large enough such that $\norm{A x}_\infty \le \rho$.)
    \item We update $x_* \gets x_* + x_t$.
    \end{enumerate}
  \item MW outputs $(1 / T) \cdot x_* \in \R^n$.
  \end{enumerate}
\end{algorithm}

\section{A Boosting Framework for Transshipment}\label{sec:boosting-framework}

We describe how to compute an $(1 + \eps)$-approximate solution for transshipment given only a black-box oracle which computes an $\alpha$-approximate dual solution. This oracle is called the dual-only $\alpha$-approximator (where the dual is defined as in \Cref{eq:TS-primal-dual}). 

\begin{definition}[Dual-Only Approximator]\label{def:preconditioner}
  Let $G \cong (B, W)$ be a weighted graph. A dual-only $\alpha$-approximator for transshipment over $G$ is an oracle which, given any proper demand $d \in \R^V$, outputs a dual solution $\phi \in \R^V$ satisfying the following properties:
  \begin{itemize}\setlength\itemsep{0em}
  \item (Dual feasibility) $\norm{ W^{-1} B^T \phi }_\infty \le 1$.
  \item (Approximation guarantee): $\inner{d, \phi} \ge \frac{1}{\alpha} \norm{d}_\opt$ .
  \end{itemize}
\end{definition}
Note that, directly from its definition, a dual-only $\alpha$-approximator can be used to obtain an $\alpha$-approximate value of the solution.

We typically want $\poly(\log n)$-approximators or $n^{o(1)}$-approximators that can be constructed and queried efficiently. However, for pedagogical purposes, we first show that the minimum spanning tree (MST) is a non-trivial $n$-approximator.
\begin{example}[MST]Let $d \in \R^V$ be an arbitrary proper demand. The MST $T$ of $G$ can be used as a simple $n$-approximator for transshipment. First, root the MST $T$ in an arbitrary $r \in V$ and assume without loss of generality (up to re-orientation of edges in $E$) that all edges point from parent to child nodes in this rooted tree. Next, let $f_T$ be the unique flow supported on $T$ that perfectly routes $d$. We now define $\phi \in \mathbb{R}^V$ by saying $\phi(r) := 0$ and proceeding in a top-to-bottom order. For each parent-child tree-edge $e = (p, c)$, set $\phi(c) := \phi(p) - \mathrm{sign}(f_T(e)) w(e)$. Now, by construction, we have $\norm{W f_T}_1 = \inner{\phi, d}$ (decompose the flow $f_T$ into a positive combination of oriented paths such that all of them cross each tree edge with the same orientation; the claim is true for each one of them and, therefore, for their sum). Furthermore, $\norm{W^{-1} B^T \phi}_\infty \le n$ by the following argument: consider each edge $e = \{u, v\} \in E_G$ and consider the unique $u-v$ tree-path. This path is composed of at most $n$ edges, and each one of them have weight at most $w(e)$ (since $T$ is the MST). Hence, $|\phi(u) - \phi(v)| \le w(e) \cdot n$, which is equivalent to the claim above. Finally, defining $\phi_* := \frac{1}{n} \phi$ we get a feasible dual solution that is $n$-approximate: Using \Cref{fact:weak-duality} we have $||d||_\opt \le \norm{W f_T}_1 = \inner{\phi, d} = n \cdot \inner{\phi_*, d}$, as required.
\end{example}



We now show the central claim of our framework: given a dual-only $\alpha$-approximator we can leverage the multiplicative weights framework (\Cref{table:feasbility-task-def}) to provide feasible potentials (i.e., a dual solution) $\phi \in \R^V$ that are $(1 + \eps)$-approximate, i.e., $\inner{d, \phi} \ge \frac{1}{1 + \eps} \norm{d}_\opt$. The existence of the boosting procedure is formalized in \Cref{lemma:primal-dual-boosting-to-eps}, while the explicit algorithm is deferred to \Cref{algo:boosting-strategy} in \Cref{sec:other-MW-forms}.

As an assumption to simplify the exposition, we can safely assume that we know the value of $g := \norm{d}_\opt$ (up to a multiplicative $1 + \eps$) as this value can be ``guessed'' via a standard binary search since our method will either certify that $\norm{d}_\opt \le (1 + \eps)g$ (telling us our guess $g$ is too low), or will otherwise construct a feasible solution $\phi$ with $\inner{d, \phi} \ge g$ (telling us we can increase our guess). 

\begin{lemma}\label{lemma:primal-dual-boosting-to-eps}
  Let $(B, W, d)$ be a transshipment instance and let $\eps > 0$. Given any $g \ge 0$, and any dual-only $\alpha$-approximator, there is a $4 \eps^{-2} \alpha^2 \ln (2m)$-round algorithm that, in each round, queries the approximator once and performs $O(1)$ basic vector operations. At termination, the algorithm either:
  \begin{enumerate}
  \item outputs (feasible) potentials $\phi_* \in \R^V$ satisfying $\norm{W^{-1} B^T \phi_*}_\infty \le 1$ and $\inner{d, \phi_*} \ge g$, or,
  \item detects that $\norm{d}_\opt \le (1 + \eps)g$; indeed, it outputs an (infeasible) flow $f_* \in \R^E$ satisfying $\norm{W f_*}_1 \le g$ and $\norm{d - Bf_*}_\opt \le \eps g$.
  \end{enumerate}
\end{lemma}
\textbf{Remark.} If one is only concerned about finding an $(1 + \eps)$-approximate dual solution, one can completely ignore the infeasible flow that is being outputted and simply use the fact that the second condition guarantees $\norm{d}_\opt \le (1 + \eps)g$, which is sufficient for binary search. Furthermore, we note that any such (infeasible) flow $f_*$ satisfying the above properties guarantees $\norm{d}_\opt \le (1 + \eps)g$ by the following argument. First, by definition of $\norm{d - B f_*}_\opt \le \eps g$, there exists a flow $f\residual$ that routes $d - B f_*$ and has cost $\norm{W f\residual}_1 \le \eps g$. Then, the flow $f_* + f\residual$ routes demand $d$ (since $B f_* + B f\residual = B f_* + d - B f_* = d$) and has cost at most $\norm{W f_*}_1 + \norm{W f\residual}_1 \le g + \eps g = (1 + \eps) g$, implying that $\norm{d}_\opt \le (1 + \eps)g$.

\begin{proof}
  First, finding potentials $\phi \in \R^V$ satisfying $\norm{W^{-1} B^T \phi}_\infty \le 1$ and $\inner{d, \phi} \ge g$ is equivalent to finding potentials $\exists? \phi\in\R^V \mid\ \norm{W^{-1} B^T \phi}_\infty - \inner{\frac{1}{g} d, \phi} \le 0$ (one direction is immediate, the other direction follows by the fact that we can scale $\phi$ such that $\inner{d, \phi} = g$). Therefore, it is sufficient to solve the following so-called TS Feasibility task (see \Cref{table:ts-feasibility-task}).

  \begin{figure}[h]
  \begin{mdframed}
    \begin{tabular}{l r l}
      \textbf{TS feasbility task}: & $\exists? \phi\in\R^V \mid$ & $\norm{W^{-1} B^T \phi}_\infty \le 1$ and $\inner{d, \phi} \ge g$. \\
      \textbf{TS feasbility task} (equivalent): & $\exists? \phi\in\R^V \mid$ & $\norm{W^{-1} B^T \phi}_\infty - \inner{\frac{1}{g} d, \phi} \le 0$. \\
      \textbf{Linearized TS task} (given $\norm{p}_1 \le 1$): & $\exists? \phi\in\R^V \mid$ & $\inner{p, W^{-1} B^T \phi} - \inner{\frac{1}{g} d, \phi} \le - \eps$. \\
      \textbf{Linearized TS task} (equivalent): & $\exists? \phi\in\R^V \mid$ & $\inner{d\residual, \phi} \ge \eps \cdot g$, \\
    & & where $d\residual := d - B (g \cdot W^{-1} p)$.
  \end{tabular}
  \caption{The (second, equivalent form of the) TS Feasibility task is a subcase of the Feasibility task defined in \Cref{table:feasbility-task-def} with $A := W^{-1} B^T, b := (1/g)d, \gamma := 0$, and renaming $x \to \phi$. The equivalent forms of the tasks follow by straightforward algebraic manipulation.%
  }
  \label{table:ts-feasibility-task}
  \end{mdframed}
  \end{figure}

  We apply the MW framework by interpreting the TS Feasbility task as a Feasibility Task in the sense of \Cref{table:feasbility-task-def}, solving it via \Cref{algo:MW-strategy-feasibility} where we have to implement the Oracle.

  First, we note that the TS Feasibility task directly corresponds to a $(A := W^{-1} B^T, b:= - \frac{1}{g} d, \gamma := 0)$-Feasibility task. We aim to implement the Oracle for the corresponding Linearized TS task with a small width $\rho$. To recap, $\rho$ is the maximum value of $\norm{W^{-1} B^T \phi}_\infty$ ever returned by the Oracle---we later determine that setting $\rho := \alpha$ suffices.

  The Oracle, upon receiving $p$ by \Cref{algo:MW-strategy-feasibility}, queries the dual-only $\alpha$-approximator with the (so-called) residual demand $d\residual := d - B (g \cdot W^{-1} p)$. Intuitively, we can interpret $p$, or more specifically $g \cdot W^{-1} p$, as the ``current'' iterate of the final flow solution. Specifically, $\norm{g \cdot W^{-1} p}_1 \le g$, i.e., it has a small cost since $g$ is a guess for $\opt$. If the residual demand can be routed with a small cost of at most $\eps g$ (which can be estimated via the approximator), it means that $\norm{d\residual} \le \eps g$, hence $f_* := g \cdot W^{-1} p$ satisfies the second output condition.

  The appoximator, being asked to route $d\residual$, returns the $\alpha$-approximate feasible dual, i.e., a vector $\phi\residual \in \R^V$ satisfying $\inner{d\residual, \phi\residual} \ge \frac{1}{\alpha} \norm{d\residual}_\opt$ and $\norm{W^{-1} B^T \phi\residual}_\infty \le 1$. The Oracle outputs $\phi' := \alpha \cdot \phi\residual$. Note that the width of the oracle is exactly $\norm{W^{-1} B \phi'}_\infty = |\alpha| \norm{W^{-1} B \phi\residual}_\infty \le \alpha \cdot 1 = \alpha$.

  Either $\inner{d\residual, \phi'} \ge \eps \cdot g$, and the Oracle successfully solves the Linearized TS task by returning $\phi\residual$, in which case the MW loop continues. If this is always the case, \Cref{algo:MW-strategy-feasibility} outputs $\phi_*$ satisfying $\norm{W^{-1} B^T \phi_*}_\infty \le 1$ and $\inner{d, \phi_*} \ge g$, as required. Regarding the width of the solution, we have that $\norm{W^{-1} B^{T} \phi'}_\infty \le \alpha$, hence setting $\rho := \alpha$ suffices, leading to at most $4\eps^{-2}\alpha^2 \ln 2m$ rounds of the algorithm.

  On the other hand, if this is (ever) not the case, we say the Oracle fails. In this case, at the moment of failure, we define $f_* := g \cdot W^{-1} p$ and have that $\inner{d\residual, \phi'} \le \eps \cdot g$. Since $\phi' = \alpha \phi\residual$, we have $\inner{d\residual, \phi\residual} \le \frac{\eps}{\alpha} \cdot g$. Since $\phi\residual$ is an $\alpha$-approximate dual w.r.t. $d\residual$, we have that $\norm{d - B(g \cdot W^{-1} p)}_\opt = \norm{d\residual}_\opt \le \alpha \cdot \frac{\eps}{\alpha} g = \eps \cdot g$. Therefore, $f_*$ satisfies the second condition and we are done.
\end{proof}
The full algorithm is given as \Cref{algo:boosting-strategy} in the Appendix.

\textbf{Reducing the residual error of the primal.} While the booster of \Cref{lemma:primal-dual-boosting-to-eps} returns a feasible $(1+\eps)$-approximate dual solution, it does not return a feasible primal solution (i.e., flow). However, this issue can be resolved by repeatedly routing the residual demand $d - B f_*$ until the cost of routing the residual demand drops to an insignificant $1/\poly(n)$-fraction of the original cost, at which point any trivial reparation scheme suffices (like routing along the MST). See \Cref{sec:reducing-residual-error} for more details. Combining the above result with binary searching the guess $g$ and with the residual error reduction (but without the model-specific trivial routing), we immediately yield the following result (proof deferred to \Cref{sec:reducing-residual-error}).
\begin{restatable}{corollary}{corSmallResidualRrror}\label{corSmallResidualRrror}
  Let $(B, W, d)$ be a transshipment instance. Given any $1/2 \ge \eps > 0$, $C > 0$ and dual-only $\alpha$-approximator, there is an $\tilde{O}(C \cdot \eps^{-2} \alpha^2)$-round algorithm computing (both):
  \begin{itemize}
  \item a feasible dual $\phi_*$ satisfying $(1 + \eps)^{-1} \norm{d}_\opt \le \inner{d, \phi} \le \norm{d}_\opt$, and,
  \item an (infeasible) primal $f_*$ satisfying $\norm{W f}_1 \le (1 + \eps) \norm{d}_\opt$ and $\norm{d - B f_*}_\opt \le n^{-C} \norm{d}_\opt$.
  \end{itemize}
  In each round, the algorithm performs $O(1)$ basic vector operations and queries to the approximator.
\end{restatable}

\section{Applications}\label{sec:applications}

In this section, we show how to apply the boosting framework of \Cref{sec:boosting-framework} to simplify and decouple several landmark results in the parallel and distributed settings. First, we describe results that approximate transshipment by solving it on a compact graph representation called a spanner or emulator (\Cref{sec:application-spanners-emulators}). Then, we describe results that use linear cost approximators (\Cref{sec:linear-cost-approximators}).

\subsection{Approximating via spanners and emulators}\label{sec:application-spanners-emulators}

A $\beta$-approximate \textbf{emulator} of a graph $G = (V, E_G)$ is a weighted graph $H = (V, E_H)$ on the same vertex set where the distances are approximated with a distortion of $\beta$; i.e., $\dist_G(u, v) \le \dist_H(u, v) \le \beta \cdot \dist_G(u, v)$ for all $u, v \in V$. A spanner is simply an emulator that is a subgraph of $G$, i.e., where $E_H \subseteq E_G$, making it particularly well-studied in some settings.

Approximating with emulators is conceptually straightforward: faced with a transshipment instance on $G$, we (approximately) solve the instance on $H$, which yields an approximate solution on $G$. This is captured by the following result.
\begin{theorem}\label{theorem:preconditioning-by-spanners}
  Let $H$ be a $\beta$-approximate emulator of $G$. Any dual-only $\alpha$-approximator on $H$ is a dual-only $(\alpha \cdot \beta)$-approximator on $H$.
\end{theorem}
\begin{proof}
  Fix a demand $d$ on $G$. Querying the $H$-approximator, we obtain a dual solution $\phi_H$ satisfying $\norm{W^{-1}_H B_H^T \phi_H}_\infty \le 1$; we also know an accompanying primal solution $f_H$ exists with $W_H(f_H) \le \alpha \inner{d, \phi_H}$.

  \textit{Primal solution.} We construct a flow $f_G$ in $G$ as follows. For each edge $e \in E_H$ we know, due to $\dist_G(u, v) \le \dist_H(u, v)$, that there exists a path in $G$ of length at most $w_H(e)$; we add $f_H(e)$ amount of flow along this path. It is easy to check that, $f_G$ routes $d$ (i.e., $B_G f_G = d$, hence it is feasible) and that $W_G(f_G) \le W_H(f_H)$, hence $\norm{d}_{\opt(G)} \le \alpha \inner{d, \phi_H}$.

  \textit{Dual solution.} Let $\phi_G := \frac{1}{\beta}\phi_H$. Note that $\norm{d}_{\opt(G)} \le (\alpha \beta) \cdot \inner{d, \phi_G}$, hence it is sufficient to deduce $\norm{W^{-1}_G B_G^T \phi_G}_\infty \le 1$. Since $\phi_H$ is feasible in $H$, we have for each $e' = \{u', v'\}\in E_H$ that $(B_H^T \phi_G)_{e'} = |\phi_G(u') - \phi_G(v')| \le \frac{|\phi_H(u') - \phi_H(v')|}{\beta} = \frac{w_H(u', v')}{\beta}$. Fix an edge $e = \{u, v\} \in E_G$; since $\dist_H(u, v) \le \beta \cdot \dist_G(u, v)$ there exists a path $(u = p'_0, p'_1, p'_2, \ldots, p'_k = v)$ in $H$ of length at most $\beta \cdot w_G(e)$. Therefore, we can deduce that $\norm{W^{-1}_G B_G^T \phi_G}_\infty \le 1$ in the following way: $|(W^{-1}_G B_G^T \phi_G)_e| = \frac{|\phi_G(u) - \phi_G(v)|}{w_G(e)} \le \frac{\sum_{i=1}^T |\phi_G(p'_{i-1}) - \phi_G(p'_i)| }{w_G(e)} \le \frac{\sum_{i=1}^T w_H(p'_{i-1}, p'_i) }{\beta w_G(e)} \le \frac{\beta w_G(e)}{\beta w_G(e)} = 1$.
\end{proof}
\textbf{Remark.} There are a few immediate extensions to the above proof. Given a primal-dual approximator (one that returns both a primal and a dual) on a \emph{spanner}, we can immediately obtain a primal-dual approximator on $G$ since the returned primal $f_H$ is also a feasible primal in $G$. A similar property holds for emulators, but one would need to provide a mapping that embeds each edge $e \in E_G$ into (paths of) $H$ that are of length at most $\beta \cdot w(e)$ in order to construct the flow $f_G$ on $G$.

\textbf{Application: TS in Broadcast congested clique~\cite{BKKL17}.} Using algorithms from prior work, a Broadcast congested clique can compute an $\tilde{O}(1)$-approximate Baswana-Sen~\cite{baswana2007simple} spanner $H$ in $\tilde{O}(1)$ rounds. The edges of such a spanner are naturally partitioned into $n$ parts of size $\tilde{O}(1)$, where each part is associated with a unique node, and that node knows the edges in its part. Therefore, the spanner can be made global knowledge in $\tilde{O}(1)$ rounds using broadcasts. Therefore, each node can solve a transshipment instance on $H$, providing an $\tilde{O}(1)$-approximator for the original graph via \Cref{theorem:preconditioning-by-spanners}, culminating in an $\tilde{O}(\eps^{-2})$-round solution for $(1+\eps)$-transshipment.

\textbf{Application: existentially-optimal SSSP in Broadcast CONGEST~\cite{BKKL17}.} Consider the single-source shortest path (SSSP) problem where each node wants to compute $(1+\eps)$-approximate from some source $s \in V$. From prior work, we can compute an \emph{overlay graph} $G' = (V', E')$ where $V' \subseteq V$ and $|V'| = \tilde{O}(\eps^{-1} \sqrt{n})$ such that the SSSP task on $G$ reduces to SSSP on $G'$, and $G'$ can be computed in $\tilde{O}(D + \eps^{-1} \sqrt{n})$ rounds. As was shown in \cite{BKKL17}, an SSSP instance can be solved by solving $\tilde{O}(1)$ transshipment instances (the details are non-trivial and out of scope of this paper), hence the problem reduces to solving TS on $G'$. However, any $T$-round Broadcast congested clique algorithm can be simulated on $G'$ in $T \cdot O(D + |V'|) = T\cdot \tilde{O}(D + \eps^{-1} \sqrt{n})$ rounds of Broadcast CONGEST: we simulate a single round by constructing a BFS tree on $G$ (of depth $O(D)$ and in $O(D)$ rounds), and then pipelining all $|V'|$ messages (that are to be broadcasted in the current round) to the root and them down to all other nodes, taking $O(D + |V'|)$ rounds in Broadcast CONGEST per round of Broadcast congested clique. Combining with the Broadcast congested clique result, we obtain an $\tilde{O}(\eps^{3})(D + \sqrt{n})$-round algorithm.

\textbf{Application: near-optimal TS in PRAM~\cite{AndoniSZ20}.} The paper introduces a concept called \emph{low-hop emulator} $H = (V, E_H)$ of $G = (V, E)$ satisfying (i) $H$ is an $\tilde{O}(1)$-approximate emulator of $G$, (ii) $|E_H| = \tilde{O}(n)$, and (iii) $\dist_H^{O(\log \log n)}(u, v) = \dist_H(u, v)$, i.e., every (exact) shortest path in $H$ has at most $O(\log \log n)$ hops (edges). Moreover, low-hop emulators can be computed in PRAM in $\tilde{O}(1)$ depth and $\tilde{O}(m)$ work. Low hop emulators are particularly useful since Property (iii) implies that one can compute (exact) SSSP on them in $\tilde{O}(1)$ depth and $\tilde{O}(n)$ work (e.g., using $O(\log \log n)$ rounds of Bellman-Ford). The ability to compute exact SSSP enables each node of $H$ to be embedded into $\ell_1$ space of dimension $\tilde{O}(1)$ with (worst-case) distortion $\tilde{O}(1)$ (via so-called Bourgain's embedding~\cite{Bou85} via $\tilde{O}(1)$ SSSP oracle calls). Since $H$ is an emulator of $G$, the same embedding is an $\tilde{O}(1)$-distortion embedding of $G$. Using \Cref{theorem:preconditioning-by-spanners}, this reduces $(1+\eps)$-TS to finding an $\tilde{O}(1)$-approximator in $\ell_1$ space. This can be done in $\tilde{O}(1)$ depth and $\tilde{O}(n)$ work using linear cost approximators (explained in \Cref{sec:linear-cost-approximators}) by utilizing the so-called \emph{randomly shifted grids} method~\cite{indyk2003fast}. This culminates in an $\tilde{O}(\eps^{-2})$ depth and $\tilde{O}(\eps^{-2} m)$ work $(1+\eps)$-transshipment algorithm.


\subsection{Approximating by linear cost approximators}\label{sec:linear-cost-approximators}

A particularly successful type of approximator for transshipment has been the linear cost approximator. The successes of such an approximator include the first $m^{1 + o(1)}$ algorithm for transshipment in the centralized model~\cite{She17b} and the first $\tilde{O}(m)$-work and $\tilde{O}(1)$-depth parallel shortest path algorithm~\cite{AndoniSZ20, Li20}.

\begin{definition}\label{def:linear-cost-approximator}
  An $\alpha$-approximate linear cost approximator for a weighted graph $G$ is a $k \times n$ matrix $P$, such that, for any proper demand $d$ it holds that $\norm{d}_\opt \le \norm{P d}_1 \le \alpha \norm{d}_\opt .$
\end{definition}


Our insight is that one can immediately convert a linear cost approximator $P$ to a dual-only approximator. Note that the $\mathrm{sign}$ function is applied entry-wise to a vector.
\begin{theorem}
  Let $P$ be an $\alpha$-approximate linear cost approximator. Consider the function $\phi(d)$ that maps a demand $d$ to $\phi(d) := \frac{1}{\alpha} P^T \mathrm{sign}(P d)$. Then, $\phi$ is a dual-only $\alpha$-approximator.
\end{theorem}
\begin{proof}
  Let $G \cong (B, W)$ be the underlying graph. First, we show that the following subclaim about a linear-algebraic guarantee that characterizes $P$: we have that $\norm{y P B W^{-1}}_\infty \le \alpha$ over all $\norm{y}_\infty \le 1$. Specifically, for each oriented edge $\vec{e} \in \vec{E}$, consider how $P$ approximates the cost of routing a unit from the head to the tail of $\vec{e}$. Formally, we define the demand $d_{\vec{e}}$ to be $d_{\vec{e}}(x) := \1{x=s} - \1{x=t}$ for an edge $\vec{e} = (s, t) \in \vec{E}$. Clearly, $\norm{d_{\vec{e}}}_\opt \le w(e)$, hence it is necessary that $\norm{P d_{\vec{e}} w(e)^{-1}}_1 \le \alpha$. Furthermore, it is easy to see that the columns of $B$ are exactly $d_{\vec{e}}$ over all $\vec{e} \in \vec{E}$, hence each column of $P B W^{-1}$ has $\ell_1$-norm at most $\alpha$. This is equivalent to $\norm{y^T P B W^{-1}}_\infty \le \alpha$ over all $\norm{y}_\infty \le 1$. This proves the subclaim.

  We now prove the complete result. Let $y := \mathrm{sign}(P d)$ and $\phi(d) := \frac{1}{\alpha} P^T y$. Since, $\norm{d}_\opt \le \norm{P d}_1$, there must exists a flow $f$ satisfying $d$ such that $\norm{d}_\opt \le \norm{W f}_1 \le \norm{P d}_1$. We verify all properties \Cref{def:preconditioner}. (Dual feasibility) $\norm{W^{-1} B^T \phi(d)}_\infty = \frac{1}{\alpha} \norm{W^{-1} B^T P^T y}_\infty \le \frac{1}{\alpha} \cdot \alpha = 1$ via the subclaim. (Approximation guarantee) $\inner{d, \phi(d)} = \frac{1}{\alpha} \inner{P d, y} = \frac{1}{\alpha} \inner{P d, \mathrm{sign}(P d)} = \frac{1}{\alpha} \norm{P d}_1 \ge \frac{1}{\alpha} \norm{d}_\opt$. \qedhere
\end{proof}

Having a dual-only $\alpha$-approximator that can be evaluated in $M$ time, we construct (via \Cref{corSmallResidualRrror}) an $\tilde{O}(\eps^{-2} \alpha^2 \cdot M)$ time $(1+\eps)$-approximate algorithm for transshipment.
\begin{corollary}\label{corollary:linear-cost-approximator-gives-algorithms}
  Let $P$ be an $\alpha$-approximate linear cost approximator on a weighted graph $G$ and suppose that we can evaluate matrix-vector multiplications with $P$ and $P^T$ (and other basic vector operations) in $M$ time. Given any TS instance, there is an $\tilde{O}(\eps^{-2} \alpha^2 M)$-time algorithm that computes a $(1 + \eps)$-approximate primal-dual pair $(f, \phi)$ satisfying the properties listed in \Cref{corSmallResidualRrror}.
\end{corollary}

\textbf{Application: almost-optimal sequential TS~\cite{She17b}.} The goal is to construct $\eps^{-2} m^{1 + o(1)}$-time $(1+\eps)$-TS solver in the sequential setting. Following \Cref{corollary:linear-cost-approximator-gives-algorithms}, it is sufficient to construct a $n^{o(1)}$-approximate linear cost approximator $P$, which is accomplished as follows. Each vertex of a weighted graph $G$ is embedded into $\ell_1$ space of dimension $O(\log^2 n)$ with (worst-case) distortion $O(\log n)$ (via so-called Bourgain's embedding~\cite{Bou85} in $\tilde{O}(m)$ sequential time). Then, the dimension of the embedding is reduced to $d := O(\sqrt{\log n})$ via a simple Johnson-Lindenstrauss projection~\cite{dasgupta1999elementary}, increasing the distortion of the embedding to $\exp(O(d)) = n^{o(1)}$. Finally, the paper constructs a $O(\log^{1.5} n)$-approximate linear cost approximator in this (virtual) $\ell_1$ space of dimension $d$ that can be evaluated efficiently, leading to a $\exp(O(d)) \cdot O(\log^{1.5} n) = n^{o(1)}$-approximate linear cost approximator in $G$, which yields the result. \textbf{Approximator in $\ell_1$ space:} We give a short cursory description on how to construct the approximator $P$. Re-scale and round the $\ell_1$ space such that all coordinates are integral. Then, each point $x$ calculates the distance $c(x)$ to the closest point with all-even coordinates. Then, $x$ uniformly spreads its demand $d(x)$ among all points with all-even coordinates that are of distance exactly $c(x)$ to $x$. Finally, repeat the algorithm on points with all-even coordinates (delete other points, divide all coordinates by $2$). After $O(\log n)$ iterations, the entire remaining demand will be supported on $2^d$ vertices of the hypercube, which can be routed to a common vertex yielding a $O(d)$ approximation. It can be shown that the cost incurred by spreading the demand at any particular step $O(d)$-approximates the optimal solution, and that the optimal solution does not increase in-between two steps, leading to a $O(d \log n) = O(\log^{1.5} n)$-approximate linear cost approximator. \textbf{Efficiency:} Evaluating the approximator requires computing the demands at each step in the above algorithm. Evaluating even the first step requires $n 2^d$ time since each point $x$ sends its demand to (potentially) $2^d = n^{o(1)}$ closest all-even points. Therefore, the dimension of the embedding is reduced to $O(\sqrt{\log n})$. Moreover, the paper (implicitly) claims this approximator in $\ell_1$ can be evaluated in $m^{1 + o(1)}$ time. Finally, we remark that the approximator does not yield a flow in the original graph in any meaningful way, (i.e., it only approximates costs), confirming that it is dual-only. Together, we solve $(1+\eps)$-TS in $\eps^{-2} m^{1 + o(1)}$ time.
  
\textbf{Application: near-optimal TS in PRAM~\cite{Li20}.} The goal is to solve $(1+\eps)$-TS in $\tilde{O}(1)$ depth and $\tilde{O}(m)$ work in PRAM. The paper constructs an $\tilde{O}(1)$-approximate linear cost approximator $P$ with sparsity $\tilde{O}(m)$, meaning it can be evaluated in $\tilde{O}(1)$ depth and $\tilde{O}(m)$ work, which would yield the result. To do so, the paper follows \cite{She17b} by embedding $G$ in $\ell_1$ space with distortion $\tilde{O}(1)$ and dimension $d := \tilde{O}(1)$ and then uses the randomly shifted grids methods of \cite{indyk2003fast} to approximate the cost in this virtual space. \textbf{Approximator in $\ell_1$ space:} We define a \emph{randomly shifted grid} of scale $W$ to be the set $W (\Z^d + u) \subseteq \R^d$, where each coordinate of $u \in \R^d$ is uniformly drawn from $[0, 1)$ (i.e., one obtains a randomly shifted grid by taking all integral $d$-dimensional points, randomly translating them along each axis, them multiplying all coordinates by $W$). Initially, set $W \gets \tilde{O}(1)$. The routing works by sampling $s := \tilde{O}(1)$ randomly shifted grids of scale $W$ and, for each grid, each point $x$ sends $1/s$ of its demand $d(x)$ to the closest point in the grid. The scale $W$ is increased by a polylogarithmic factor and the algorithm is repeated for $O(\log n)$ steps until all demand is supported on a hypercube, at which point it can be $O(d)$-approximated by aggregating it at a single vertex. It can be shown that the cost incurred by routing the demand at any particular step $\tilde{O}(1)$-approximates the optimal solution, and that the optimal solution increases only by a multiplicative $1 + 1/\poly(\log n)$ factor, hence after $O(\log n)$ iterations we obtain an $\tilde{O}(1)$-approximate linear cost approximator $P$ that has sparsity $\tilde{O}(m)$. \textbf{Vertex reduction framework:} On its face, the above approach simply shows that in order to get $(1 + \eps)$-transshipment (and $(1 + \eps)$-shortest paths, as arduously shown in the paper), it is sufficient to find an $\tilde{O}(1)$-distortion $\ell_1$-embedding. However, to find an $\ell_1$-embedding, one needs $\tilde{O}(1)$-approx shortest paths (with some additional technical requirements concerning the violation of the triangle inequality). To resolve this cycle, the paper goes through the vertex reduction framework of \cite{Mad10,Pen16} which reduces the number of vertices by a polylogarithmic factor, recursively solves transshipment, lifts the solution to the original graph, and repairs it using the boosting framework. The details are out-of-scope.

\textbf{Future work.} The ideas used for solving transshipment have historically paralleled the ideas used for solving maximum flow problems. Adding to the connection between these two problems, approximate solutions to maximum flow can also be boosted in a similar way to transshipment~\cite{She13} via linear cost approximators (called \emph{congestion approximators}). However, no framework that can handle black-box approximators has been developed---creating such a framework would conceptually simplify the task of designing approximate maximum flow solutions. Furthermore, both transshipment and maximum flow are special cases of the so-called $\ell_p$-norm flow, which also seems to support boosting~\cite{adil2019iterative}. 




\bibliography{Refs}


\appendix

\section{A Simple Analysis of Multiplicative Weights (MW)}\label{sec:simple-mw-analysis}

In this section, we exhibit a particularly simple analysis of multiplicative weights. We first define a natural optimization task in \Cref{sec:canonical-optimization}, provide an algorithm and its analysis in \Cref{sec:analysis-MW}, and then use it to solve other tasks (like the Feasibility task from \Cref{table:feasbility-task-def}) in \Cref{sec:other-MW-forms}.

Our analysis forgoes the typical explanation that goes through the weighted majority (also known as the \emph{experts}) algorithm and accompanying ad-hoc calculations~\cite{AroraHK12}. Instead, we show how to relax an often-found (non-smooth) optimization task into a smooth one by replacing the (non-smooth) maximum with a well-known \emph{smooth max} (or log-sum-exp) function (defined in \Cref{def:smax}). Then, we show that multiplicative weights can be seen as an instance of Frank-Wolfe method~\cite{frank1956algorithm} adjusted to optimizing the smooth maximum function over a convex set by maintaining a dual over the probability simplex. Using well-known elementary properties of the smooth max, this approach yields a particularly simple analysis of the algorithm. While it is entirely possible that this perspective was known to experts in the area, the author is not aware of any write-up providing a similar analysis.

\subsection{Solving an optimization task using MW}\label{sec:canonical-optimization}

In this section, we define the so-called \emph{Canonical optimization task}, from which we will derive solutions to all other tasks.

\textbf{Notation.} We define $\norm{x}_{\max} = \max_i x_i$ to be the largest coordinate of a vector, an $\Delta_m := \{ x \in \R^m \mid x \ge 0, \sum_{i=1}^m x_i = 1 \}$ be the set of $m$-element probability distributions (the so-called probability simplex).

\begin{definition}
  Let $K$ be an arbitrary convex subspace $K \subseteq \R^n$, $A \in \R^{m \times n}$ be a matrix, and $b \in \R^n$ be a vector. We define the following \textbf{Canonical optimization task}:
  \begin{align*}
    \min_{x \in K} & \norm{Ax}_{\max} + \inner{b, x} .
  \end{align*}
  Given additionally a vector $p \in \Delta_m$ we define the accompanying \textbf{Linearized canonical (optimization) task}:
  \begin{align*}
    \min_{x \in K} \inner{p, Ax} + \inner{b, x} .
  \end{align*}
\end{definition}
Note that for each $x \in K$ we have $\inner{p, Ax} + \inner{b, x} \le \norm{Ax}_{\max} + \inner{b, x}$, hence the Linearized task is a relaxation of the optimization task.

Suppose we want to solve some fixed Canonical optimization task and assume we know how to solve the accompanying (typically much easier!) Canonical linearized task (for any $p$ and $\eps$) via some \emph{black-box Oracle}. Then, there exists a simple solver that computes a solution to the Canonical optimization task by repeatedly querying the Oracle with different values of $p \in \Delta_m$ (the accuracy $\eps$ stays fixed); the oracle is assumed to return a feasible solution $x \in K$ for each queried Linearized canonical task.

We define the \textbf{width of the Oracle} $\rho > 0$ to be (any upper bound on) the largest width of a solution $\norm{Ax}_{\infty}$ that can be returned by the Oracle, i.e., $\rho \ge \norm{A x}_{\infty}$ during the algorithm. Oracles with larger widths need to be queried more times, hence we aim to construct Oracles with their width being as small as possible. The solver is given in \Cref{algo:smooth-MW-strategy} and its properties are stated in \Cref{theorem:canonical-MW}.

\begin{algorithm}[h]
  \caption{Solver for the Canonical optimization task.}
  \label{algo:smooth-MW-strategy}
  \begin{enumerate}
  \item \textbf{Input:} Canonical optimization task ($A \in \R^{m \times n}, b \in \R^n, \gamma \in R$) and $\eps > 0$.
  \item \textbf{Definition:} $\nabla \smax_{\beta}(x) := \left( \frac{ \exp(\beta \cdot x_i) }{\sum_{j=1}^m \exp(\beta \cdot x_j)} \right)_{i=1}^m \in \mathbb{R}^m$ (See \Cref{sec:analysis-MW}).
  \item Initialize $x_* \gets \vec{0} \in \R^n$ and $\beta := \eps / (2 \rho^2)$.
  \item For $t = 1, \ldots, T$ rounds, where $T := 4\eps^{-2} \rho^2 \ln m$:
    \begin{enumerate}
    \item Let $p_t \gets [\nabla \smax_{\beta}] ( A x_* ) \in \Delta_m$.
    \item MW outputs $p_t \in \Delta_m$ to Oracle.
    \item Oracle returns a solution $x_t \in \mathbb{R}^n;\ \norm{A x_t}_\infty \le \rho$ to the Canonical linearized optimization task w.r.t. $p_t$.
    \item We update $x_* \gets x_* + x_t$.
    \end{enumerate}
  \item MW outputs $(1 / T) \cdot x_* \in \R^n$.
  \end{enumerate}
\end{algorithm}

  

\begin{theorem}\label{theorem:canonical-MW}
  Let $(A, b)$ be a Canonical optimization task and fix $\eps > 0$. Suppose we have access to an Oracle that will solve the accompanying Canonical linearized task specified by $(A, b, p, \eps)$ for any $p \in \Delta_m$. Then, \Cref{algo:smooth-MW-strategy} constructs a solution $x' \in K$ satisfying $\norm{Ax'}_{max} + \inner{b, x'} \le \min_{x} \norm{Ax}_{max} + \inner{b, x} + \eps$. During the construction, \Cref{algo:smooth-MW-strategy} queries the Oracle at most $4 \eps^{-2} \rho^2 \ln (2m)$ times. Here, $\rho > 0$ is the width of the Oracle.
\end{theorem}


\subsection{Analysis of the canonical MW algorithm}\label{sec:analysis-MW}

On a high-level, we will solve the Canonical optimization task by relaxing it to the so-called \emph{Smooth optimization task} by replacing the $\max$ with the so-called \emph{smooth maximum} $\smax_\beta$. We introduce the $\smax$ function and state its properties.
\begin{fact}\label{def:smax}
  We define $\smax_\beta : \R^m \to \R$ as
  $$\smax_\beta(x) = \frac{1}{\beta} \ln\left( \sum_{i=1}^m \exp(\beta x_i) \right),$$
  where $\beta > 0$ is some \emph{accuracy} parameter (increasing $\beta$ increases accuracy but decreases smoothness). The following properties holds:

  \begin{enumerate}
  \item \label{prop:smax-max} The maximum is approximated by $\smax$:
    $$\smax_{\beta}(x) \in \left[ \norm{x}_{\max}, \norm{x}_{\max} + \frac{\ln n}{\beta} \right] .$$
  \item \label{prop:smax-grad-distrib} The gradient of $\smax$ is some probability distribution over $[n]$:
    $$\nabla\! \smax_{\beta}(x) = ( \frac{1}{Z} \exp(\beta \cdot x_i) )_{i=1}^m \in \Delta_m,$$ where $Z := \sum_{i=1}^n 
    \exp(\beta \cdot x_i)$ is the normalization factor.
  \item \label{prop:smax-smooth} $\smax_{\beta}$ is convex and $\beta$-smooth with respect to $\norm{\cdot}_\infty$:
    $$\smax_\beta(x + h) \le \smax_\beta(x) + \inner{\nabla\! \smax_{\beta}(x), h} + \beta\cdot \norm{h}_\infty^2$$
  \item \label{prop:smax-zero} $\smax_{\beta}( \vec{0} ) = \frac{\ln m}{\beta}$.
  \end{enumerate}
\end{fact}
The stated properties of $\smax_\beta$ are elementary and can be directly verified (e.g., see \cite{She13,BKKL17}). For instance, Property~\ref{prop:smax-smooth} is equivalent to verifying that the Hessian $\nabla^2 \smax_\beta$ satisfies $0 \le x^T (\nabla^2 \smax_\beta) x \le 2 \beta \inner{x, x}$ for all $x \in \R^m$.

We are now ready to introduce the \emph{Smooth optimization task} and its linearization.
\begin{table}[h]
  \begin{tabular}{l r r}
    \textbf{Smooth optimization task}: & $\min_{x \in K}$ & ${\color{blue}\smax}_\beta(Ax) + \inner{b, x}$ \\
    \textbf{Linearized smooth task}, given $x_* \in K$: & $\min_{x \in K}$ & $\inner{\nabla [\smax_{\beta}(Ax_*)], Ax} + \inner{b, x}$
  \end{tabular}
\end{table}
We first note that solving the Smooth optimization task is harder than solving the Canonical optimization task since $\smax_\beta(Ax) + \inner{b, x} \ge \norm{Ax}_{\max} + \inner{b, x}$. Furthermore, it uses only smooth functions, hence we can use tools from calculus to analyze its value. It is important to note that the Linearized smooth task is exactly the Linearized canonical task after substituting $p \gets [\nabla \smax_\beta](A x_*) \in \Delta_m$ (i.e., the gradient of $\smax_{\beta}$, evaluated at $A x_*$).

We now prove the efficacy of \Cref{algo:smooth-MW-strategy} for its ability to solve the Canonical optimization task.
\begin{proof}[Proof of \Cref{theorem:canonical-MW}]
  We track the objective value of the Smooth optimization task via $\Phi(x) := \smax_\beta( A x ) + \inner{b, x}$. We remind the reader that $\Phi(x)$ is a pessimistic estimator of the Canonical optimization task, hence bounding $\Phi(x)$ suffices for the canonical task. For future reference, note that $\nabla \Phi(x) = A^T [\nabla \smax_{\beta}](A x) + b$.

  In the $i^{th}$ step, let $x_{-} = \sum_{i=1}^{t-1} x_i$ be $x$ at the start of the step, and let $x_{+} = x_{-} + x_t$ be $x$ after the step's update. In each step we have:
  \begin{align*}
    \Phi(x_{+}) & = \Phi(x_{-} + x_t) = \smax_\beta(A (x_{-} + x_t)) + \inner{b, x_{-} + x_t} \\
                & \le \Phi(x_{-}) + \inner{\nabla \Phi(x_{-}), x_t} + \beta \norm{A x_t}_\infty^2 + \inner{b, x_t} && \text{(\Cref{def:smax}.\ref{prop:smax-smooth})} \\
                & = \Phi(x_{-}) + \left\{\inner{A^T [\nabla \smax_{\beta}](Ax_{-}), x_t} + \inner{b, x_t}\right\} + \beta \norm{A x_t}_\infty^2\\
                & = \Phi(x_{-}) + \left\{\inner{p_t, A x_t} + \inner{b, x_t}\right\} + \beta \norm{A x_t}_\infty^2\\
                & = \Phi(x_{-}) + \mathrm{LinearizedTaskValue}_t + \eps/2
  \end{align*}

  By assumption, the value of the smooth Linearized task was always at most $\mu^* := \max_{t=1}^T \mu_t$. Therefore, applying the above single-step analysis for $T$ steps, we get that the final value $x_*$ satisfies $\Phi(x_*) \le \Phi(\vec{0}) + \sum_{t=1}^T ( \mu^* + \eps/2 ) = \frac{\ln m}{\beta} + T \cdot (\mu^* + \eps/2) \le T \cdot (\mu^* + \eps)$. The last inequality holds when $T \ge \frac{\ln m}{(\eps/2) \beta} = 4 \eps^{-2} \rho^2 \ln m$ and we have that $\frac{\ln m}{\beta} \le T \cdot (\eps/2)$.

  The algorithm's output $(1/T) \cdot x_* \in K$ since it can be written as $(1/T) \sum_{i=1}^T x_t$, an average of $T$ vectors in $K$. Furthermore, since $\norm{A x_*}_{\max} + \inner{b, x_*} \le \Phi(x_*) \le T \cdot (\mu^* + \eps)$, we have that $\norm{A x_*/T}_{\max} + \inner{b, x_*/T} \le \mu^* + \eps$, as required.
\end{proof}

\subsection{Deriving other forms of the MW algorithm}\label{sec:other-MW-forms}

In this section, we use the Canonical optimization task to solve other tasks, namely, the Feasibility task (\Cref{table:feasbility-task-def}) and provide pseudocode for the TS Feasibility task (\Cref{table:ts-feasibility-task}).

\textbf{Solving the feasbility task.}
The main idea of the derivation is that we can convert between $\norm{\cdot}_\infty$ and $\norm{\cdot}_{\max}$ via the following identity: $\norm{Ax}_\infty = \norm{\twovec{A}{-A} x}_{\max}$, which enables us to leverage \Cref{theorem:canonical-MW} to prove \Cref{theorem:MW-feasibility}. We first state Algorithm~\ref{algo:MW-strategy-feasibility} and then prove the result.


\begin{proof}[Proof of \Cref{theorem:MW-feasibility} and Algorithm~\ref{algo:MW-strategy-feasibility}]
  We apply \Cref{theorem:canonical-MW} to the Canonical optimization task $\norm{\twovec{A}{-A} x}_{\max} + \inner{b, x}$, which is paired up with the the linearized canonical task
  $$\inner{\twovec{p_1}{p_2}, \twovec{A}{-A}x} + \inner{b, x} = \inner{p_1 - p_2, A x} + \inner{b, x}.$$
  Note that $\twovec{p_1}{p_2} \in \Delta_{2m}$ implies $\norm{p_1 - p_2}_1 \le 1$. 

  We directly obtain an algorithm in which the Oracle, upon being given $p' := p_1 - p_2$ (with $\norm{p'}_1 \le 1$), either returns a solution $x \in \R^n$ satisfying $\inner{p', Ax} + \inner{b, x} \le \gamma - \eps$ (i.e., the Linearized task from \Cref{theorem:MW-feasibility}), or we say the Oracle fails. If the oracle never fails, MW computes a solution $x_* \in \R^n$ satisfying $\norm{Ax}_{\max} + \inner{b, x} \le \gamma$, which provides a solution for the Feasibility task, as required.

  Finally, the width $\rho_{Thm.\ref{theorem:canonical-MW}}$ of the Oracle is the maximum $\norm{[Ax, -Ax]}_{\max}$ that can ever be returned. Therefore, we can assign $\rho := \rho_{Thm.\ref{theorem:canonical-MW}}$ as $\norm{[Ax, -Ax]}_{\max} = \norm{A x}_\infty$, hence the same value $\rho$ is an upper bound on the width with respect to \Cref{theorem:MW-feasibility}. Therefore, the number of rounds of the game is $4 \eps^{-2} \rho^2 \ln (2m)$, as required.

  Finally, we confirm that Algorithm~\ref{algo:MW-strategy-feasibility} is simply Algorithm~\ref{algo:smooth-MW-strategy} with the correct substitutions and with a manual computation of the gradient $\nabla \smax$.
\end{proof}

\textbf{Pseudocode for TS Feasibility task (i.e., the boosting algorithm).} One can easily verify that combining the proof of \Cref{lemma:primal-dual-boosting-to-eps} with Algorithm~\ref{algo:MW-strategy-feasibility}, yields the following Algorithm~\ref{algo:boosting-strategy} for boosting a transshipment solution.

\begin{restatable}{algorithm}{algoBoostingStrategy}
  \caption{Boosting a dual-only $\alpha$-approximator.}
  \label{algo:boosting-strategy}
  \begin{enumerate}
  \item \textbf{Input:} transshipment instance ($B, W, d$), current guess $g > 0$, $\eps > 0$, and an $\alpha$-approximator.
  \item Initialize $\phi_* \gets \vec{0} \in \R^n$ and $\beta := \eps / (2 \rho^2)$.
  \item For $t = 1, \ldots, T$ rounds, where $T := 4 \eps^{-2} \alpha^2 \ln (2m)$:
    \begin{enumerate} 
    \item Let $q \gets \begin{bmatrix} W^{-1} B^T \phi_* \\ - W^{-1} B^{T} \phi_* \end{bmatrix} \in \R^{2m}$.
    \item Let $q'_i \gets \exp(\beta q_i)_i$ for $i \in [2m]$.
    \item Let $p_t \gets (1 / \sum_i q'_i)(q'_i - q'_{i+m})$.
    \item Let $f_t \gets g \cdot W^{-1} p_t \in \R^{\vec{E}}$ be a flow with cost $\norm{W f_t}_1 \le g$.
    \item Algorithm queries the approximator with the demand $d_{\mathrm{residual}} \gets d - B f_t$.
    \item Approximator finds $\phi\residual$ with either:
      \begin{enumerate}
      \item $\inner{d\residual, \alpha \cdot \phi\residual} \ge \eps \cdot g$ (in which case we continue), or,
      \item (otherwise) we stop and output $f_t$; this guarantees $\norm{W f_t}_1 \le g \norm{p}_1 = g$ and $\norm{d - B f_t}_\opt \le \norm{d\residual}_\opt \le \eps g$.
      \end{enumerate}
    \item Update $\phi_* \gets \phi_* + \alpha \cdot \phi\residual$.
    \end{enumerate}
  \item Output potentials $(1 / T) \cdot \phi_* \in \R^V$. The output satisfies $\inner{d, (1/T) \phi_*} \ge g$.
  \end{enumerate}
\end{restatable}

\section{Reducing the Residual Error}\label{sec:reducing-residual-error}

A downside of using dual-only approximators is that the primal solution $f$ returned by \Cref{lemma:primal-dual-boosting-to-eps} is not feasible. It does not perfectly satisfy $Bf = d$, but merely that that the residual flow $d - Bf$ can be routed with small cost, i.e., $\norm{d - Bf}_\opt \le \norm{d}_\opt$. Following \cite{She17b}, the issue can be partially ameliorated by applying the same procedure $\tilde{O}(1)$ times: in each step we route the residual demand of the previous step and combine the outputs together. This has the effect of reducing the cost of routing the residual demand to an $(n^{-C})$-fraction of the original (for any constant $C > 0$), as shown by the result below.

\corSmallResidualRrror*
\begin{proof}
  \Cref{lemma:primal-dual-boosting-to-eps} provides the following. Given a demand $d$ we can obtain a primal-dual pair $(f, \phi)$ such that $(1 + \eps)^{-1} \cdot \norm{d}_\opt \le \inner{d, \phi} \le \norm{W f}_1 \le \norm{d}_\opt$, and $\norm{d - B f}_\opt \le \eps \norm{d}_\opt$.

  Let $d_0 := d$. For $i \in \{0, 1, \ldots, T\}$ where $T := \tilde{O}(C)$, we apply \Cref{lemma:primal-dual-boosting-to-eps} on the demand $d_i$ to get $(f_{i}, \phi_{i})$ such that $\norm{d_{i+1}}_\opt \le \eps/2 \cdot \norm{d_i}_\opt$, where $d_{i+1} := d_i - B f_{i}$. We have that $\norm{d_i}_\opt \le (\eps/2)^i \norm{d}_\opt$, hence $\norm{W f_i}_1 \le (\eps/2)^i \norm{d}_\opt$.

  Let $f_* := f_0 + f_1 + \ldots + f_T$ and $\phi_* = \phi_0$. It immediately follows that $\phi_*$ is feasible and satisfies the required conditions.

  We now verify that the cost of $f_*$ is $(1 + \eps)$-approximate. We have that
  \begin{align*}
    \norm{W f_*}_1 \le \sum_{i=0}^T \norm{W f_i}_1 \le \norm{d}_\opt \cdot(1 + \eps/2 + (\eps/2)^2 + \ldots + (\eps/2)^T) \le (1 + \eps) \norm{d}_\opt .
  \end{align*}

  Finally, we verify the cost of routing the residual demand. First, $d - d_{T+1} = \sum_{i=0}^T (d_{i} - d_{i+1}) = \sum_{i=0}^T B f_i = B f_*$, hence $\norm{d - B f_*}_\opt = \norm{d_{T+1}} \le (\eps/2)^{T+1} \norm{d}_\opt \le n^{-C} \norm{d}_\opt$. This completes the proof.
\end{proof}

This reduction is often sufficient to recover a feasible $(1+\eps)$-approximate primal solution by using some trivial $\poly(n)$-approximate way to route the residual demand. For instance, routing along the minimum spanning tree (MST) yields an $n$-approximation to $\norm{d}_\opt$. Therefore, finding a residual demand $d'$ with $\norm{d'}_\opt \le n^{-2} \norm{d}_\opt$ and routing $d'$ along the MST yields the flow $f$ of cost $(1 + \eps)\norm{d}_\opt + n \cdot n^{-2} \norm{d}_\opt \le (1 + 2 \eps)\norm{d}_\opt$~\cite{She17b}.
\end{document}